\documentclass[runningheads]{llncs}
\usepackage{graphicx}
\usepackage{amsmath}
\usepackage{amssymb}
\newcommand{\mc}{{\rm mc}}
\newcommand{\cmc}{{\rm cmc}}
\newcommand{\mb}{{\rm mb}}

\begin{document}
\title{An Improved Fixed-Parameter Algorithm for Max-Cut Parameterized by Crossing Number\thanks{This work is partially supported by JSPS KAKENHI Grant Numbers JP16H02782, JP16K00017, JP16K16010, JP17H01788, JP18H04090, JP18H05291, JP18K11164, and JST CREST JPMJCR1401.}}
\titlerunning{An FPT Algorithm for Max-Cut Parameterized by Crossing Number}
% If the paper title is too long for the running head, you can set
% an abbreviated paper title here
%
\author{Yasuaki~Kobayashi\inst{1}\and
Yusuke~Kobayashi\inst{2} \and
Shuichi~Miyazaki\inst{3}\and
Suguru~Tamaki\inst{4}}
\authorrunning{Y. Kobayashi et al.}
% First names are abbreviated in the running head.
% If there are more than two authors, 'et al.' is used.
%
\institute{Graduate School of Informatics, Kyoto University \and Research Institute for Mathematical Sciences, Kyoto University\and
Academic Center for Computing and Media Studies, Kyoto University \and
School of Social Information Science, University of Hyogo}
\maketitle % typeset the header of the contribution
\begin{abstract}
The Max-Cut problem is known to be NP-hard on general graphs, while it can be solved in polynomial time on planar graphs.  
In this paper, we present a fixed-parameter tractable algorithm for the problem on ``almost'' planar graphs:
Given an $n$-vertex graph and its drawing with $k$ crossings, our algorithm runs in time $O(2^k(n+k)^{3/2} \log (n + k))$.
Previously, Dahn, Kriege and Mutzel (IWOCA 2018) obtained an algorithm that, given an $n$-vertex graph and its $1$-planar drawing with $k$ crossings, runs in time $O(3^k n^{3/2} \log n)$.
Our result simultaneously improves the running time and removes the $1$-planarity restriction.
\end{abstract}
%
%
%
% ===============================
\section{Introduction}\label{sec:intro}

The Max-Cut problem is one of the most basic graph problems in theoretical computer science.  In this problem, we are given an edge-weighted graph, and asked to partition the vertex set into two parts so that the total weight of edges having endpoints in different parts is maximized.  This is one of the 21 problems shown to be NP-hard by Karp's seminal work~\cite{Karp72}.
To overcome this intractability, numerous researches have been done from the viewpoints of approximation algorithms~\cite{GW95,KKMO07,Has01,TSSW00}, exponential-time exact algorithms~\cite{GS17,W05}, and fixed-parameter algorithms~\cite{CJM15,MSZ18,MR99,RS07}.
There are several graph classes for which the Max-Cut problem admits polynomial time algorithms \cite{BJ00,Gur99}.
Among others, one of the most remarkable tractable classes is the class of planar graphs.
Orlova and Dorfman~\cite{OD72} and Hadlock~\cite{Had75} developed polynomial time algorithms for the unweighted Max-Cut problem on planar graphs, which are subsequently extended to the weighted case by Shih et al.~\cite{SWK90} and Liers and Pardella~\cite{LP12}.

Dahn et al.\@~\cite{DKM18} recently presented a fixed-parameter algorithm for 1-planar graphs.  A graph is called {\em 1-planar} if it can be embedded into the plane so that each edge crosses at most once. Their algorithm runs in time $O(3^k n^{3/2} \log n)$, where $n$ is the number of vertices and $k$ is the number of crossings of a given 1-planar drawing.  Their algorithm is a typical branching algorithm: at each branch, it removes a crossing by yielding three sub-instances. After removing all of the $k$ crossings, we have at most $3^{k}$ Max-Cut instances on planar graphs.  Each of these problems can be solved optimally in $O(n^{3/2} \log n)$ time by reducing to the maximum weight perfect matching problem with small separators \cite{LT80}, thus giving the above mentioned time complexity.

\paragraph{\bf Our contributions.}
%In this paper, we first show the NP-hardness of the unweighted Max-Cut problem on 1-planar graphs.

%\begin{theorem}\label{thm:hardness}
% The Max-Cut problem is NP-hard even for unweighted 1-planar graphs.
%\end{theorem}

%Next, we give an improved fixed-parameter algorithm which is faster than Dahn et al.~\cite{DKM18}'s.

%\begin{theorem}\label{thm:alg}
%    Given a graph $G$ and its 1-planar drawing with $k$ crossings,
%    the Max-Cut problem can be solved in $O(2^k \poly(n))$ time.
%\end{theorem}

%Similarly to Dahn et al.~\cite{DKM18}'s, our algorithm branches to remove a crossing, but it yields not three but only two sub-instances. However, in compensation for this advantage, these two sub-instances are not purely Max-Cut instances but with some condition.  To achieve our goal, we modify the reduction of \cite{LP12} so that our ``conditional'' Max-Cut problem can be reduced to the maximum weight perfect matching problem.

%We remark that in the proof of Theorem \ref{thm:hardness}, we reduce the Max-Cut problem on general graphs (with a given drawing in the plane) to that on 1-planar graphs with a 1-planar drawing, without changing the number of crossings.  Thus our algorithm is applicable not only to 1-planar drawings but to {\em any} drawings:

%\begin{corollary}\label{coro:alg}
%    Given a graph $G$ and its drawing with $k$ crossings,
%    the Max-Cut problem can be solved in $O(2^k \poly(n))$ time.
%\end{corollary}

To the best of the authors' knowledge, it is not known whether the Max-Cut problem on 1-planar graphs is solvable in polynomial time.  In this paper, we show that it is NP-hard even for unweighted graphs.

\begin{theorem}\label{thm:hardness}
 The Max-Cut problem on unweighted 1-planar graphs is NP-hard even when a 1-planar drawing is given as input.
\end{theorem}

Next, we give an improved fixed-parameter algorithm, which is the main contribution of this paper:

\begin{theorem}\label{thm:alg}
    Given a graph $G$ and its drawing with $k$ crossings,
    the Max-Cut problem can be solved in $O(2^k(n+k)^{3/2}\log (n + k))$ time.
\end{theorem}

Note that our algorithm not only improves the running time of Dahn et al.'s algorithm \cite{DKM18}, but also removes the $1$-planarity restriction.  An overview of our algorithm is as follows: Using a polynomial-time reduction in the proof of Theorem \ref{thm:hardness}, we first reduce the Max-Cut problem on general graphs (with a given drawing in the plane) to that on 1-planar graphs with a 1-planar drawing, without changing the number of crossings.  We then give a faster fixed-parameter algorithm than Dahn et al.~\cite{DKM18}'s for Max-Cut on 1-planar graphs.

The main idea for improving the running time is as follows: Similarly to Dahn et al.~\cite{DKM18}, we use a branching algorithm, but it yields not three but only two sub-instances. Main drawbacks of this advantage are that these two sub-instances are not necessarily on planar graphs, and not necessarily ordinary Max-Cut instances but with some condition, which we call the ``constrained Max-Cut'' problem.  To solve this problem, we modify the reduction of \cite{LP12} and reduce the constrained Max-Cut problem to the maximum weight $b$-factor problem, which is known to be solvable in polynomial time in general \cite{Gab83}. We investigate the time complexity of the algorithm in \cite{Gab83}, and show that it runs in $O((n + k)^{3/2} \log (n + k))$-time in our case, which proves the running time claimed in Theorem \ref{thm:alg}.

%For Step~2, we show the following:
%\begin{theorem}\label{thm:alg}
%    Given a graph $G$ and its 1-planar drawing with $k$ crossings,
%    the Max-Cut problem can be solved in $O(2^k \poly(n))$ time.
%\end{theorem}

\paragraph{\bf Independent work.}
Chimani et al. \cite{CDJKKMN19} independently and simultaneously achieved the same improvement by giving an $O(2^k(n+k)^{3/2}\log (n + k))$ time algorithm for the Max-Cut problem on embedded graphs with $k$ crossings. They used a different branching strategy, which yields $2^k$ instances of the Max-Cut problem on planar graphs. 

\paragraph{\bf Related work.}
The Max-Cut problem is one of the best studied problems in several areas of theoretical computer science.
This problem is known to be NP-hard even for co-bipartite graphs \cite{BJ00}, comparability graphs \cite{Poc16}, cubic graphs \cite{Yan78}, and split graphs \cite{BJ00}.
From the approximation point of view, the best known approximation factor is $0.878$ \cite{GW95}, which is tight under the Unique Games Conjecture \cite{KKMO07}.
In the parameterized complexity setting, there are several possible parameterizations for the Max-Cut problem. 
Let $k$ be a parameter and let $G$ be an unweighted graph with $n$ vertices and $m$ edges.
The problems of deciding if $G$ has a cut of size at least $k$ \cite{RS07}, at least $m - k$ \cite{PPW16}, at least $m / 2 + k$ \cite{MR99}, $m / 2 + (n-1)/4 + k$ \cite{CJM15}, or at least $n - 1 + k$ \cite{MSZ18} are all fixed-parameter tractable. 
For sparse graphs, there are efficient algorithms for the Max-Cut problem. It is well known that the Max-Cut problem can be solved in $O(2^{t}tn)$ time \cite{BJ00}, where $t$ is the treewidth of the input graph.
The Max-Cut problem can be solved in polynomial time on planar graphs \cite{Had75,LP12,OD72,SWK90}, which has been extended to bounded genus graphs by Galluccio et al. \cite{GLV01}.
%The Max-Cut problem for unweighted graphs is fixed-parameter tractable parameterized by the genus of the input graph \cite{GLV01}. This result can be extended to graphs with integral weights bounded by a constant (**Check**). Since there is a graph that can be embedded into a surface of genus $1$ and has an arbitrary number of crossings, our parameterization is rather restrictive. However, the restriction on edge weights is indispensable for the result of \cite{GLV01}. The parameterized complexity of the Max-Cut problem (with no restriction on edge weights) of bounded genus graphs is still open.

%\paragraph{\bf Organization of the paper.} The paper is organized as follows.

% ===============================
\section{Preliminaries}\label{sec:prel}

%Throughout this paper, we work on edge weighted graphs.
In this paper, an edge $\{u, v\}$ is simply denoted by $uv$, and a cycle $\{v_0, v_1, \ldots, v_k\}$ with edges $\{v_i, v_{(i+1)\mod k}\}$ for $0 \le i \le k$ is denoted by $v_0v_1\ldots v_k$.

A graph is {\em planar} if it can be drawn into the plane without any edge crossing.
A {\em crossing} in the drawing is a non-empty intersection of edges distinct from their endpoints.
If we fix a plane embedding of a planar graph $G$, then the edges of $G$ separate the plane into connected regions, which we call {\em faces}.
%The unbounded region of the drawing is an {\em outer face}.
%Let $F$ be the set of faces in the drawing.
%The {\em dual graph} $G^*= (V^*, E^*)$ is a planar multigraph where each vertex in $V^*$ corresponds to each face in $F$ and if two faces of $G$ share an edge $e$,
%the two vertices of $G^{*}$ associated with these faces are adjacent by an edge $e^*$.
%Thus the relation ``$*$'' gives a one-to-one correspondence between $E$ and $E^*$.
%For $L \subseteq E$, we write $L^*$ to denote the set $\{e^* : e \in L\}$.
%Let us note that $G^*$ has no self-loops since we assume that $G$ has no degree-one vertices.
%For a weighted planar graph $G = (V, E, w)$, the weight function $w^*$ of the dual graph $G^{*}$ is defined accordingly: $w^*(e^*) = w(e)$ for each $e \in E$.

Consider a drawing not necessarily being planar, where no three edges intersect at the same point.
We say that a drawing is {\em 1-planar} if every edge is involved in at most one crossing.
A graph is {\em 1-planar} if it admits a 1-planar drawing.
Note that not all graphs are 1-planar: for example, the complete graph with seven vertices does not admit any 1-planar drawing.

Let $G = (V, E, w)$ be an edge weighted graph with $w: E \rightarrow \mathbb{R}$.
%A {\em cut} of $G$ is a bipartition $(S, V \setminus S)$ of $V$.
A {\em cut} of $G$ is a pair of vertex sets $(S, V \setminus S)$ with $S \subseteq V$. 
We denote by $E(S, V \setminus S)$ the set of edges having one endpoint in $S$ and the other endpoint in $V \setminus S$.
We call an edge in $E(S, V \setminus S)$ a {\em cut edge}.
The {\em size} of a cut $(S, V \setminus S)$ is the sum of the weights of cut edges, i.e., $\sum_{e \in E(S, V \setminus S)} w(e)$.
The Max-Cut problem asks to find a maximum size of a cut, denoted $\mc(G)$, of an input graph $G$.
We assume without loss of generality that $G$ has no degree-one vertices since such a vertex can be trivially accommodated to either side of the bipartition so that its incident edge contributes to the solution.  Therefore, we can first work with removing all degree-one vertices, and after obtaining a solution, we can put them back optimally.

An instance of the {\em constrained Max-Cut} problem consists of an edge weighted graph $G = (V, E, w)$, together with a set $C$ of pairs of vertices of $V$.
%An element of $C$ is called a {\em constrained pair}.
A feasible solution, called a {\em constrained cut}, is a cut in which all the pairs in $C$ are separated.
The {\em size} of a constrained cut is defined similarly as that of a cut.
The problem asks to find a maximum size of a constrained cut of $G$, denoted $\cmc(G)$.

Let $G=(V, E, w)$ be an edge weighted graph with $w: E \rightarrow \mathbb{R}$ and let $b : V \rightarrow \mathbb{N}$.
A {\em $b$-factor} of $G$ is a subgraph $H=(V, F)$ such that $F \subseteq E$ and every vertex $v$ has degree exactly $b(v)$ in $H$.
The {\em cost} of a $b$-factor $H=(V, F)$ is the sum of the weights of edges in $F$, i.e., $\sum_{e \in F} w(e)$.
The maximum weighted $b$-factor problem asks to compute a maximum cost of a $b$-factor of $G$, denoted by $\mb(G)$.
This problem can be seen as a generalization of the maximum weight perfect matching problem and is known to be solvable in polynomial time \cite{Gab83}.

% ===============================
\section{NP-Hardness on 1-Planar Graphs}\label{sec:hard}
In this section, we prove Theorem~\ref{thm:hardness}, i.e., NP-hardness of the unweighted Max-Cut problem on 1-planar graphs.
The reduction is performed from the unweighted Max-Cut problem on general graphs.  
Since we use this reduction in the next section for weighted graphs, we exhibit the reduction for the weighted case.  When considering unweighted case, we may simply let $w(e)=1$ for all $e$.

\begin{proof}[of Theorem \ref{thm:hardness}]
Fix an edge weighted graph $G = (V, E, w)$.
For an edge $e=uv$ of $G$, define a path $P_e$ consisting of three edges $uu'$, $u'v'$, and $v'v$, each of weight $w(e)$, where $u'$ and $v'$ are newly introduced vertices (see Fig.~\ref{fig:subdivide}).
\begin{figure}[ht]
  \centering
  \includegraphics[width=0.63\textwidth]{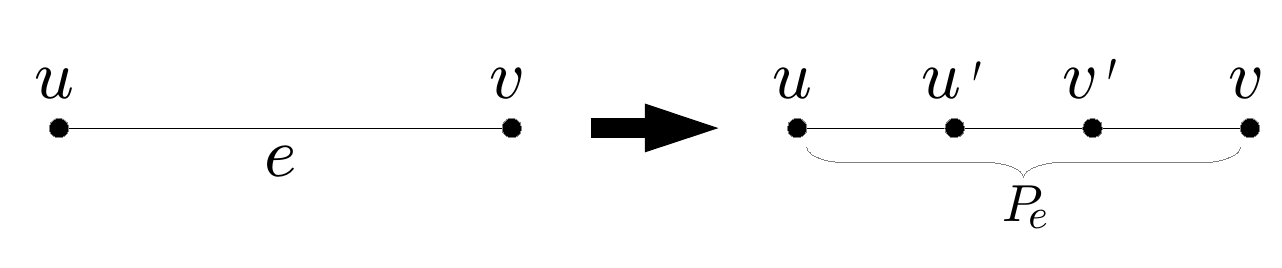}
  \caption{Replacing an edge $e$ with a path $P_e$.}\label{fig:subdivide}
\end{figure}

Let $G'$ be the graph obtained from $G$ by replacing $e$ by $P_e$. 
The following lemma is crucial for our reduction.

\begin{lemma}\label{lem:reduction}
   $\mc(G') = \mc(G) + \max(0, 2w(e))$.
\end{lemma}

\begin{proof}
Suppose first that $w(e) \ge 0$. Consider a maximum size cut $(S, V \setminus S)$ of $G$. If $u$ and $v$ are in the same side of the partition, we extend the cut to obtain a cut of $G'$ by putting $u'$ and $v'$ into the other side. Otherwise, we put $u'$ to $v$'s side and $v'$ to $u$'s side.  In both cases, the cut size increases by exactly $2w(e)$, so $\mc(G') \ge \mc(G) + 2w(e)$.
Conversely, let $(S', V' \setminus S')$ be a maximum size cut of $G'$, where $V' = V \cup \{u', v'\}$.
If $u$ and $v$ are in the same side, at least one of $u'$ and $v'$ must be in the other side, as otherwise, we can increase the size of the cut by moving $u'$ or $v'$ to the other side, contradicting the maximality of $(S', V' \setminus S')$.
This implies that exactly two edges of the path $P_e$ contribute to the cut $(S', V' \setminus S')$. 
Similarly, if $u$ and $v$ are in the different side, we can see that every edge of $P_e$ contribute to $(S', V' \setminus S')$.
Thus, the cut $(S' \setminus \{u', v'\}, (V'\setminus S') \setminus \{u', v'\})$ of $G$ is of size $\mc(G') - 2w(e)$ and hence we have $\mc(G) \ge \mc(G') - 2w(e)$.  From the above two inequalities, we have that $\mc(G') = \mc(G) + 2w(e)$.

Suppose otherwise that $w(e) < 0$. Similarly to the first case, we extend a maximum cut $(S, V \setminus S)$ of $G$ to a cut of $G'$.  
This time, we can do so without changing the cut size, implying that $\mc(G') \ge \mc(G)$: If $u$ and $v$ are in the same side, we put both $u'$ and $v'$ into the same side as $u$ and $v$.
Otherwise, we put $u'$ in $u$'s side and $v'$ in $v$'s side.
For the converse direction, let $(S', V' \setminus S')$ be a maximum size cut of $G'$.
If $u$ and $v$ are in the same (resp. different) side of the cut, the maximality of the cut implies that no (resp. exactly one) edge of $P_e$ contributes to the cut.  Hence the cut $(S' \setminus \{u', v'\}, (V'\setminus S') \setminus \{u', v'\})$ of $G$ has size $\mc(G')$, so that $\mc(G) \ge \mc(G')$.  Therefore $\mc(G') = \mc(G)$, completing the proof.
\qed
\end{proof}

Now, suppose we are given a Max-Cut instance $G$ with its arbitrary drawing.
We will consider here a crossing as a pair of intersecting edges.
We say that two crossings are {\em conflicting} if they share an edge, and the shared edge is called a {\em conflicting edge}.
With this definition, a drawing is 1-planar if and only if it has no conflicting crossings. 

Suppose that the drawing has two conflicting crossings $\{e, e'\}$ and $\{e, e''\}$ with the conflicting edge $e$ (see Fig.~\ref{fig:redraw}).
Replace $e$ by a path $P_e$ defined above and locally redraw the graph as in Fig.~\ref{fig:redraw}.
Then, this conflict is eliminated, and by Lemma~\ref{lem:reduction}, the optimal value increases by exactly $2w(e)$.
Note that this operation also works for eliminating two conflicting crossings caused by the same pair of edges.

\begin{figure}[ht]
  \centering
  \includegraphics[width=0.63\textwidth]{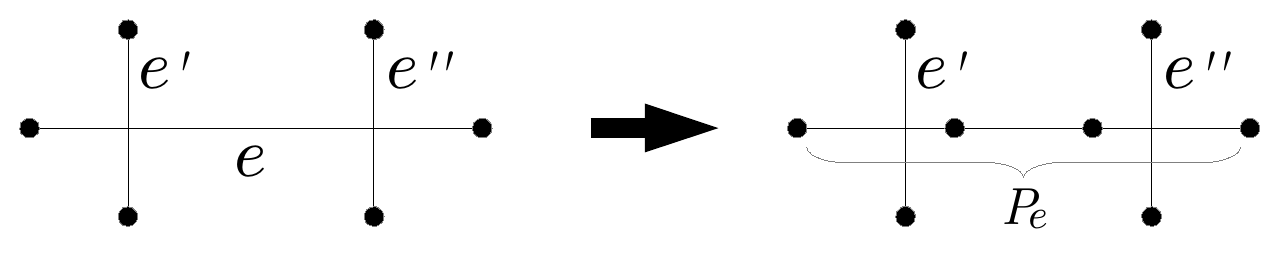}
  \caption{Eliminating a conflicting crossing.}\label{fig:redraw}
\end{figure}

We repeat this elimination process as long as the drawing has conflicting crossings, eventually obtaining a 1-planar graph.
From a maximum cut of the obtained graph, we can obtain a maximum cut of the original graph $G$ by simply replacing each path $P_e$ by the original edge $e$.
The reduction described above is obviously done in polynomial time.
Since the Max-Cut problem on general graphs is NP-hard, this reduction implies Theorem~\ref{thm:hardness} and hence the proof is completed.
\qed
\end{proof}

Note that the above reduction is in fact parameter-preserving in a strict sense, that is, the original drawing has $k$ crossings if and only if the reduced 1-planar drawing has $k$ crossings.

%Therefore, by combining the result of Dahn et al.~\cite{DKM18}, we immediately have the following FPT algorithm.

%\begin{lemma}\label{lem:slow-fpt}
%  Given a graph of $n$ vertices together with its drawing whose crossing number is $k$, the max-cut problem can be solved in $3^kn^{O(1)}$ time.
%\end{lemma}

% ===============================
\section{An Improved Algorithm}\label{sec:alg}
In this section, we prove Theorem~\ref{thm:alg}.  As mentioned in Sec.~\ref{sec:intro}, we first reduce a general graph to a 1-planar graph using the polynomial-time reduction given in the previous section.  
Recall that this transformation does not increase the number of crossings.
Hence, to prove Theorem~\ref{thm:alg}, it suffices to provide an $O(2^k(n+k)^{3/2}n)$ time algorithm for 1-planar graphs and its 1-planar drawing with at most $k$ crossings.

\subsection{Algorithm}
Our algorithm consists of the following three phases.

\paragraph{\bf Preprocessing.} 
We first apply the following preprocessing to a given graph $G$: For each crossing $\{ ac, bd\}$, we apply the replacement in Fig.~\ref{fig:subdivide} twice, once to $ac$ and once to $bd$ (and take the cost change in Lemma~\ref{lem:reduction} into account) (see Fig. \ref{fig:preprocess}). 
As a result of this, for each crossing, all the new vertices, $a'$, $b'$, $c'$ and $d'$, concerned with this crossing have degree two and there is no edge among these four vertices.
This preprocessing is needed for the subsequent phases.

\begin{figure}[ht]
  \centering
  \includegraphics[width=0.6\textwidth]{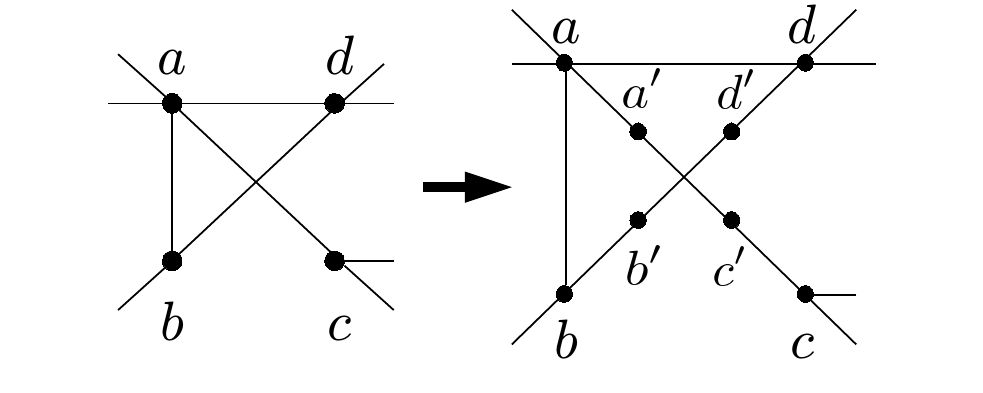}
  \caption{Preprocessing}\label{fig:preprocess}
\end{figure}

\paragraph{\bf Branching.}
As Dahn et al.~\cite{DKM18}'s algorithm, our algorithm branches at each crossing and yields sub-instances. Consider a crossing $\{ ac, bd\}$. 
%Then obviously an optimal solution satisfies one of the following two conditions: (1') $|S \cap \{a, c\}| \neq 1$ and (2') $|S \cap \{a, c\}|=1$.
%The first case can be treated in the same way as before: Delete the edge $ac$ from $G$.
%For the second case, since every solution satisfying (2') includes the edge $ac$, we delete $ac$ from $G$ and add $w(ac)$ to a solution obtained from this branch.
%But the subproblem in this branch is no longer the ordinary Max-Cut problem, since $a$ and $c$ must be separated.  We call the problem with such a constraint the {\em constrained Max-Cut} problem.  
Obviously, any optimal solution lies in one of the following two cases
(1) $|S \cap \{a, b\}| \neq 1$ and (2) $|S \cap \{a, b\}| = 1$.
To handle case (1), we construct a sub-instance by contracting the pair $\{a, b\}$ into a single vertex. 
For case (2), we add four edges $ab, bc, cd$, and $da$ of weight zero (see Fig. \ref{fig:m1}).  
Thanks to the preprocessing phase, adding these four edges does not create a new crossing.
Note that these edges do not affect the size of any cut.  These edges are necessary only for simplicity of the correctness proof.
We then add the constraint that $a$ and $b$ must be separated.  
Therefore, the subproblem in this branch is the {\em constrained} Max-Cut problem. 
Note also that in this branch, we do not remove the crossing.
We call the inner region surrounded by the cycle $abcd$ a {\em pseudo-face}, and the edge $ab$ (that must be a cut edge) a {\em constrained edge}.
Note that the better of the optimal solutions of the two sub-instances coincides the optimal solution of the original problem.
After $k$ branchings, we obtain $2^k$ constrained Max-Cut instances.

\begin{figure}[ht]
  \centering
  \includegraphics[width=0.5\textwidth]{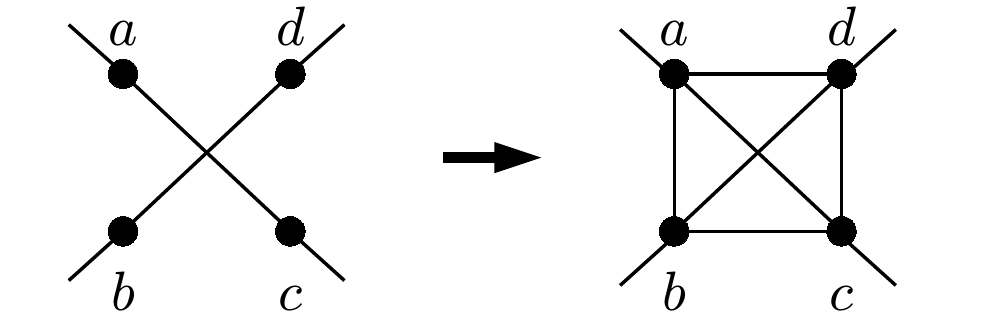}
  \caption{Adding four edges}\label{fig:m1}
\end{figure}

\paragraph{\bf Solving the Constrained Max-Cut Problem.}
In this last phase, we solve $2^k$ constrained Max-Cut instances obtained above, and output the best solution among them.
To solve each problem, we reduce it to the maximum weighted $b$-factor problem (see Sec.~\ref{sec:prel} for the definition), which is shown in Sec.~\ref{ssec:run-time} to be solvable in $O((n + k)^{3/2}\log (n + k))$ time in our case.
Hence the whole running time of our algorithm is $O(2^k(n+k)^{3/2}\log (n + k))$. 

Let $G$ be a graph (with a drawing) obtained by the above branching algorithm.
If there is a face with more than three edges, we triangulated it by adding zero-weight edges without affecting the cut size of a solution.
% If there is a face with more than three edges, we triangulated it by adding zero-weight edges without affecting the cut size of a solution (see Fig. \ref{fig:m2}). 
By doing this repeatedly, we can assume that every face of $G$ (except for pseudo-faces) has exactly three edges. 
Recall that, in the preprocessing phase, we subdivided each crossing edge twice.
Due to this property, no two pseudo-faces share an edge.

% \begin{figure}[ht]
%   \centering
%   \includegraphics[width=0.57\textwidth]{m2.pdf}
%   \caption{Triangulating a face}\label{fig:m2}
% \end{figure}
Let $G=(V,E,w)$ be an instance of the constrained Max-Cut problem.  We reduce it to an instance $(G^{*}=(V^{*},E^{*},w^{*}), b)$ of the maximum weighted $b$-factor problem.
The reduction is basically constructing a dual graph.  
For each face $f$ of $G$, we associate a vertex $f^{*}$ of $G^{*}$.  Recall that $f$ is surrounded by three edges, say $xy, yz$ and $zx$.  Corresponding to these edges, $f^{*}$ has the three edges $(xy)^{*}, (yz)^{*}$ and $(zx)^{*}$ incident to vertices corresponding to the three neighborhood faces or pseudo-faces (see Fig. \ref{fig:m3}).  The weight of these edges are defined as $w^{*}((xy)^{*})=w(xy)$, $w^{*}((yz)^{*})=w(yz)$, and $w^{*}((zx)^{*})=w(zx)$.  We also add a self-loop $l$ with $w^{*}(l)=0$ to $f^{*}$.  Note that putting this self-loop to a $b$-factor contributes to the degree of $f^{*}$ by 2.  Finally, we set $b(f^{*})=2$.  (In case some edge surrounding $f$ is shared with a pseudo-face, we do some exceptional handling, which will be explained later.)

\begin{figure}[ht]
  \centering
  \includegraphics[width=0.55\textwidth]{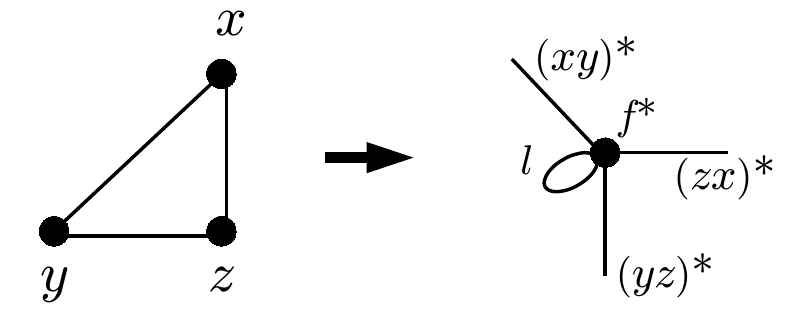}
  \caption{Reduction for a face}\label{fig:m3}
\end{figure}

For each pseudo-face $f$ of $G$, we associate a vertex $f^{*}$.  Corresponding to the three edges $bc$, $cd$, and $da$, we add the edges $(bc)^{*}, (cd)^{*}$ and $(da)^{*}$ to $E^{*}$.  (Note that we do not add an edge corresponding to $ab$.)  We also add a self-loop $l$ to $f^{*}$ (see Fig. \ref{fig:m4}).  The weight of these edges are defined as follows:
\begin{eqnarray*}
    w^*((bc)^*) &=& \frac{\beta - 2\alpha}{3}, \\
    w^*((cd)^*) &=& \frac{\alpha + \beta}{3}, \\
    w^*((da)^*) &=& \frac{\alpha - 2\beta}{3},\\
    w^*(l) &=& \frac{2\alpha + 2\beta}{3},
\end{eqnarray*}
where $\alpha = w(ac)$ and $\beta = w(bd)$.  We set $b(f^{*})=3$.  

\begin{figure}[ht]
  \centering
  \includegraphics[width=0.63\textwidth]{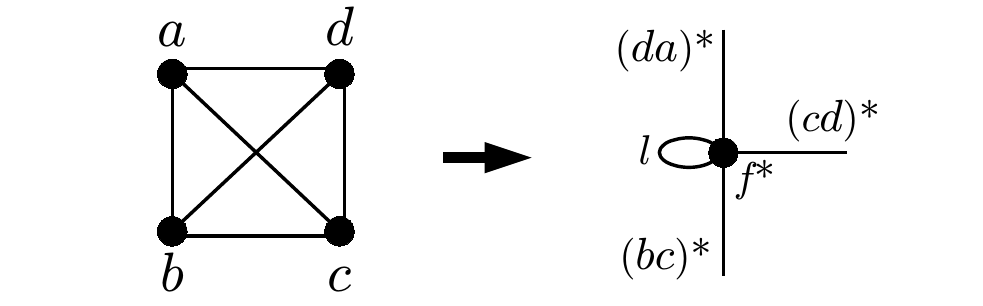}
  \caption{Reduction for a pseudo-face}\label{fig:m4}
\end{figure}

Now we explain an exception mentioned above.  Consider a (normal) face $f$ consisting of three vertices $x$, $y$, and $z$, and let $f^{*}$ be the vertex of $G^{*}$ corresponding to $f$. 
Suppose that some edge, say $xy$, is shared with a pseudo-face $g$, whose corresponding vertex in $G^{*}$ is $g^{*}$.
In this case, the edge $(xy)^{*}$ in $E^{*}$, connecting $f^{*}$ and $g^{*}$, is defined according to the translation rule for the pseudo-face $g$.
Specifically, if the edge $xy$ is identical to the edge $bc$ in Fig.~\ref{fig:m4}, then the weight $w^{*}((xy)^{*})$ is not $w(xy)$ but $\frac{\beta - 2\alpha}{3}$.
When $xy$ is identical to $cd$ or $da$ in Fig.~\ref{fig:m4}, $w^{*}((xy)^{*})$ is defined in the similar manner.

In case $xy$ is identical to the constrained edge $ab$ in Fig.~\ref{fig:m4}, the rule is a bit complicated: First, we do not add an edge $(xy)^{*}$ to $E^{*}$ (which matches the absence of $(ab)^{*}$ in Fig.~\ref{fig:m4}).
Next, we subtract one from $b(f^{*})$; hence in this case, we have $b(f^{*})=1$ instead of the normal case of $b(f^{*})=2$.  
As one can see later, the absence of $(xy)^{*}$ and subtraction of $b(f^{*})$ implicitly mean that $(xy)^{*}$ is already selected as a part of a $b$-factor, which corresponds to the constraint that $a$ and $b$ must be separated in any constrained cut of $G$.
Here we stress that this subtraction is accumulated for boundary edges of $f$.
For example, if all the three edges $xy$, $yz$, and $zx$ are the constrained edges of (different) pseudo-faces, then we subtract three from $b(f^{*})$, which results in $b(f^{*})=-1$ (of course, this condition cannot be satisfied at all and hence the resulting instance has no feasible $b$-factor).

Now we have completed the construction of $(G^{*}=(V^{*},E^{*},w^{*}), b)$.
We then show the correctness of our algorithm in Sec.~\ref{ssec:correct}, and evaluate its running time in Sec.~\ref{ssec:run-time}.

\subsection{Correctness of the Algorithm}\label{ssec:correct}

To show the correctness, it suffices to show that the reduction in the final phase preserves the optimal solutions.
Recall that $\cmc(G)$ is the size of a maximum constrained cut of $G$, and $\mb(G^{*})$ is the cost of the maximum $b$-factor of $G^{*}$.

\begin{lemma}
$\cmc(G)=\mb(G^{*})$.
\end{lemma}

\begin{proof}
We first show $\cmc(G) \leq \mb(G^{*})$.  Let $S$ be a maximum constrained cut of $G$ with cut size $\cmc(G)$.  We construct a $b$-factor $H$ of $G^{*}$ with cost $\cmc(G)$.  Informally speaking, this construction is performed basically by choosing dual edges of cut edges.  

Formally, consider a face $f$ surrounded by edges $xy, yz$ and $zx$. 
First suppose that none of these edges are constrained edges.
Then by construction of $G^{*}$, the degree of $f^{*}$ is 5 (including the effect of the self-loop) and $b(f^{*})=2$.
It is easy to see that zero or two edges among $xy, yz$ and $zx$ are cut edges of $S$.
In the former case, we add only the self-loop $l$ to $H$.
In the latter case, if the two cut edges are $e$ and $e'$, then we add corresponding two edges $(e)^{*}$ and $(e')^{*}$ to $H$.
Note that in either case, the constraint $b(f^{*})=2$ is satisfied.
%Obviously, the sum of the weights of the selected edges is the same as the sum of the weights of cut edges.

Next, suppose that one edge, say $xy$, is a constrained edge.  In this case, the degree of $f^{*}$ is 4 and $b(f^{*})=1$.
Since $xy$ is a constrained edge, we know that $xy$ is a cut edge of $S$.
Hence exactly one of $yz$ and $zx$ is a cut edge.
If $yz$ is a cut edge, we add $(yz)^{*}$ to $H$; otherwise, we add $(zx)^{*}$ to $H$.

If two edges, say $xy$ and $yz$, are constrained edges, we have that $b(f^{*})=0$.  In this case, we do not select edges incident to $f^{*}$.

Finally, suppose that all the three edges are constrained edges.  Clearly it is impossible to make all of them cut edges.  Thus $G$ admits no constrained cut, which contradicts the assumption that $S$ is a constrained cut.

Next, we move to pseudo-faces.  For each pseudo-face $f$ with a cycle $abcd$ where $ab$ is a constrained edge, we know that $a$ and $b$ are separated in $S$.  There are four possible cases, depicted in Fig.~\ref{fig:cuts}, where vertices in the same side are labeled with the same color, and bold edges represent cut edges.  Corresponding to each case of Fig.~\ref{fig:cuts}, we select edges in $G^{*}$ as shown in Fig.~\ref{fig:cf} and add them to $H$.  Note that in all four cases, the constraint $b(f^{*})=3$ is satisfied.

\begin{figure}[th]
  \centering
  \includegraphics[width=0.85\textwidth]{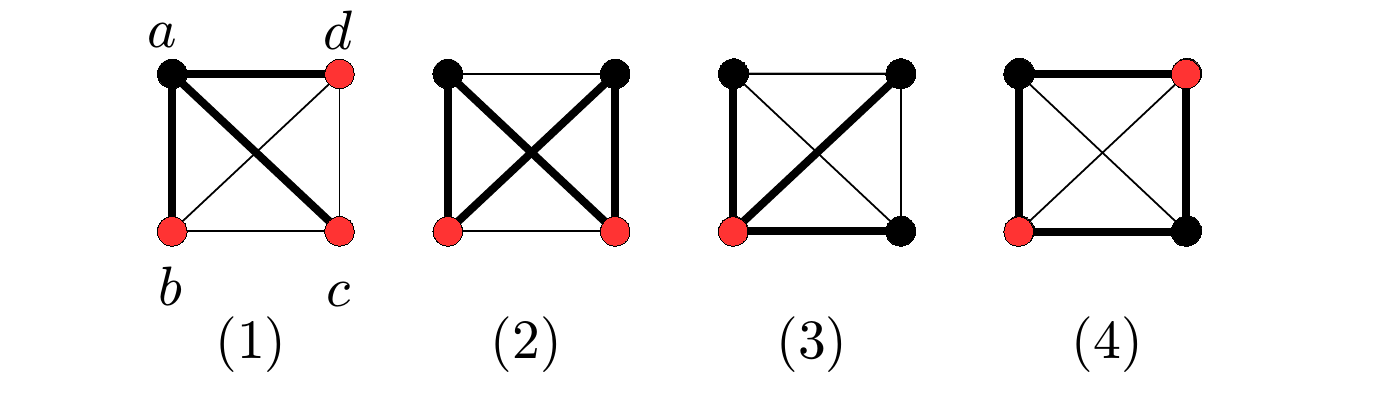}
  \caption{Feasible cuts of $G$ in which $a$ and $b$ are separated}\label{fig:cuts}
\end{figure}

\begin{figure}[th]
  \centering
  \includegraphics[width=0.85\textwidth]{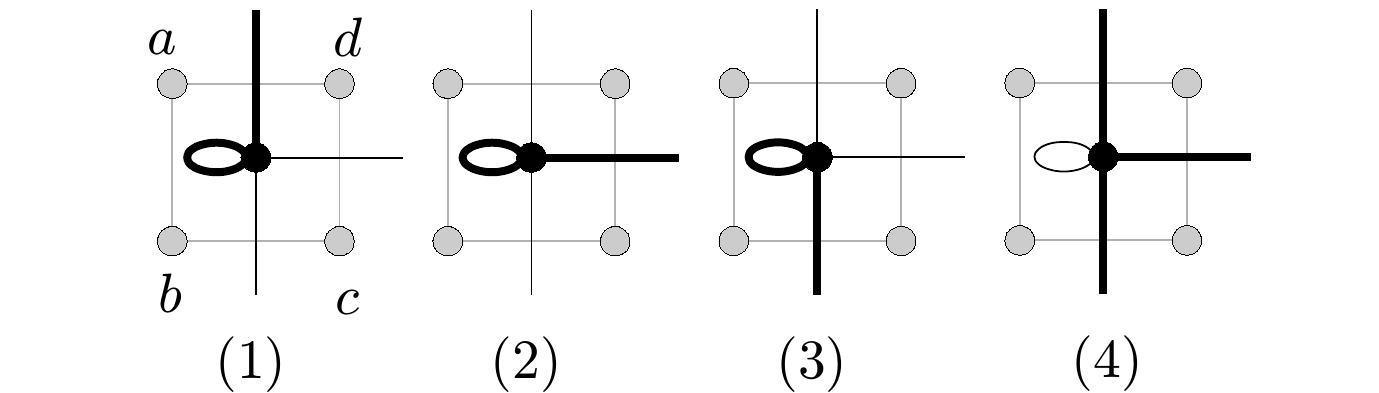}
  \caption{Solution for $G^{*}$ constructed from each cut of Fig.~\ref{fig:cuts}}\label{fig:cf}
\end{figure}

We have constructed a subgraph $H$ of $G^{*}$ and shown that for any vertex of $G^{*}$, the degree constraint $b$ is satisfied in $H$.  Hence $H$ is in fact a $b$-factor.

It remains to show that the cost of $H$ is $\cmc(G)$.  
The edges of $G$ are classified into the following two types; (1) edges on the boundary of two (normal) faces and (2) edges constituting pseudo-faces (i.e., those corresponding to one of six edges of $K_{4}$ in Fig.~\ref{fig:m4} left).
For a type (1) edge $e$, $e$ is a cut edge if and only if its dual $(e)^{*}$ is selected in $H$, and $w(e)=w^{*}((e^{*}))$.
For type (2) edges, we consider six edges corresponding to one pseudo-face simultaneously.
Note that there are four feasible cut patterns given in Fig.~\ref{fig:cuts}, and we determine edges of $H$ according to Fig.~\ref{fig:cf}.
We examine that the total weight of cut edges and that of selected edges coincide in every case:

\smallskip
\noindent
{\bf Case (1):} The weight of cut edges is $w(da)+w(ac)+w(ab)=w(ac)$.  The weight of selected edges is $w^{*}((da)^{*})+w^{*}(l) = \frac{\alpha - 2\beta}{3}+\frac{2\alpha + 2\beta}{3}=\alpha=w(ac)$.  

\smallskip
\noindent
{\bf Case (2):} The weight of cut edges is $w(ab)+w(ac)+w(bd)+w(cd)=w(ac)+w(bd)$.  The weight of selected edges is $w^{*}((cd)^{*})+w^{*}(l) = \frac{\alpha + \beta}{3}+\frac{2\alpha + 2\beta}{3}=\alpha+\beta=w(ac)+w(bd)$.  

\smallskip
\noindent
{\bf Case (3):} The weight of cut edges is $w(ab)+w(bd)+w(bc)=w(bd)$.  The weight of selected edges is $w^{*}((bc)^{*})+w^{*}(l) = \frac{\beta - 2\alpha}{3}+\frac{2\alpha + 2\beta}{3}=\beta=w(bd)$.

\smallskip
\noindent
{\bf Case (4):} The weight of cut edges is $w(ab)+w(bc)+w(cd)+w(da)=0$.  The weight of selected edges is $w^{*}((bc)^{*})+w^{*}((cd)^{*})+w^{*}((da)^{*})=\frac{\beta - 2\alpha}{3} + \frac{\alpha + \beta}{3} + \frac{\alpha - 2\beta}{3}=0$.

\smallskip
\noindent
Summing the above equalities over the whole graphs $G$ and $G^{*}$, we can conclude that the constructed $b$-factor $H$ has cost $\cmc(G)$.

To show the other direction $\cmc(G) \geq \mb(G^{*})$, we must show that, from an optimal solution $H$ of $G^{*}$, we can construct a cut $S$ of $G$ with size $\mb(G^{*})$.  Since the construction in the former case is reversible, we can do the opposite argument to prove this direction; hence we will omit it here. \qed
\end{proof}

\subsection{Time Complexity of the Algorithm}\label{ssec:run-time}

As mentioned previously, to achieve the claimed running time, it suffices to show that a maximum weight $b$-factor of $G^*$ can be computed in $O((n+k)^{3/2}\log (n + k))$ time.
The polynomial time algorithm presented by Gabow \cite{Gab83} first reduces the maximum weighted $b$-factor problem to the maximum weight perfect matching problem, and then solves the latter problem using a polynomial time algorithm.  We follow this line but make a careful analysis to show the above mentioned running time.

For any vertex $v$ of $G^*$, we can assume without loss of generality that $b(v) \le d(v)$, where $d(v)$ is the degree of $v$, as otherwise $G^*$ obviously does not have a $b$-factor.
Gabow's reduction \cite{Gab83} replaces each vertex $v$ with a complete bipartite graph $K_{d(v), d(v) - b(v)}$ as shown in Fig.~\ref{fig:b-factor}, where newly introduced edges have weight zero.
Since the maximum degree of $G^*$ is at most five, each vertex is replaced by a constant sized gadget.
Also, since $G^{*}$ is a planar graph with $O(n+k)$ vertices, the created graph has $O(n+k)$ vertices and admits a balanced separator of size $O(\sqrt{n+k})$.  
It is easy to see that this reduction can be done in $O(n+k)$ time, and $G^*$ has a $b$-factor if and only if the created graph has a perfect matching of the same weight.

\begin{figure}[th]
  \centering
    \includegraphics[width=0.48\textwidth]{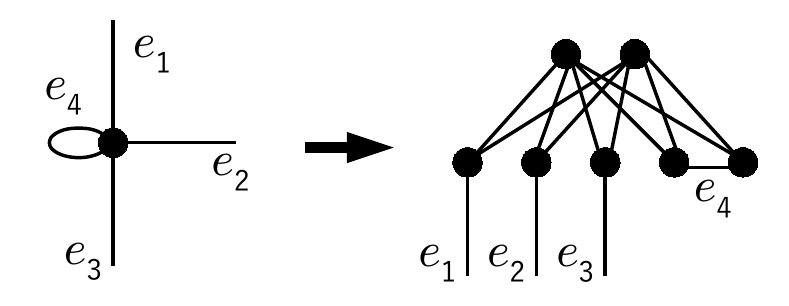}
%  \caption{Each vertex $v$ of $G^*$ is replaced with a complete bipartite graph $K_{d(v), d(v) - b(v)}$. If the vertex has a self loop, the corresponding edge is between two vertices in the complete bipartite graph.}\label{fig:b-factor}
  \caption{Gabow's reduction}\label{fig:b-factor}
\end{figure}

Lipton and Tarjan \cite{LT80} present an algorithm for the maximum weight perfect matching problem that runs in time $O(n^{3/2} \log n)$ for an $n$-vertex graph having a balanced separator of size $O(\sqrt{n})$.
Thus, by using it, we obtain the claimed running time of $O((n+k)^{3/2}\log (n + k))$.

\section{Concluding Remarks}\label{sec:conc}

In this paper, we have proposed an algorithm for the Max-Cut problem on ``almost'' planar graphs. Our algorithm runs in $O(2^k(n+k)^{3/2}\log (n + k))$ time, where $n$ is the number of vertices of an input graph and $k$ is the number of crossings of a given drawing, which improves the previous result of Dahn et al. \cite{DKM18}.

We remark that our algorithm can be applied to the NAE-SAT problem.  In this problem, given a CNF formula, a clause is satisfied if and only if it has at least one true literal and at least one false literal. The NAE-SAT problem asks to decide if the CNF formula has a satisfying assignment.
This problem is known to be NP-complete in general~\cite{Sch78}, but Moret \cite{Mor88} showed that the Planar NAE-SAT problem, the restriction of inputs to formulas whose incidence graphs are planar, is solvable in polynomial time.
To solve it, he demonstrated a polynomial time reduction from the Planar NAE-SAT problem to the Max-Cut problem on planar graphs. 
This reduction actually preserves the number of crossings, that is, a CNF formula with a drawing of its incidence graph having $k$ crossings is reduced to a Max-Cut instance with $k$ crossings.
Therefore, using our algorithm, the NAE-SAT problem can be solved in $O(2^k(n+k)^{3/2}\log (n + k))$ time, where $n$ is the number of variables in the formula and $k$ is the number of crossings in the (given) incidence graph.

%There are several open questions arisen from our results. 
%Our result improves the previous running time of Dahn et al \cite{DKM18}.
%It would be interesting to know the possibility of further improvement.
%In particular, a lower bound based on the Exponential Time Hypothesis is a possible direction for future work.
%We are also interested in the kernelizability of this parameterization.
%Finally, the Max-Cut problem for unweighted graphs is fixed-parameter tractable parameterized by the genus of the input graph \cite{}. This %result can be extended into graphs with integral weights bounded by a constant (**Check**). Since there is a graph that can be embedded into a %surface of genus $1$ and has an arbitrary number of crossings, our parameterization is rather restrictive. However, the restriction on edge %weights is essential in the result of \cite{}. The parameterized complexity of the Max-Cut problem (with no restriction on edge weights) of %bounded genus graphs is still open.

\section*{Acknowledgements}
 The authors deeply thank anonymous referees for giving us valuable comments. In particular, one of the referees pointed out a flaw in an early version of Lemma~\ref{lem:reduction}, which has been fixed in the current paper.

%
% ---- Bibliography ----
%
% BibTeX users should specify bibliography style 'splncs04'.
% References will then be sorted and formatted in the correct style.
%
% \bibliographystyle{splncs04}
% \bibliography{mybibliography}

\begin{thebibliography}{99}

%\bibitem{Bar83}
%Barahona, F.: The max-cut problem on graphs not contractible to $K_5$.
%Operation Research Letters \textbf{2}(3): 107--111 (1983)

\bibitem{BJ00}
Bodlaender, H. L., Jansen, K.: On the complexity of the maximum cut problem. Nordic Journal of Computing \textbf{7}(1): 14--31 (2000) 

\bibitem{CDJKKMN19}
Chimani, M., Dahn, C., Juhnke-Kubitzke, M., Kriegem, N. M., Mutzel, P., Nover, A.:
Maximum Cut Parameterized by Crossing Number.
arXiv:1903.06061 (2019)

\bibitem{CJM15}
Crowston, R., Jones, M., Mnich, M.: Max-Cut Parameterized Above the Edwards-Erd\"{o}s Bound. Algorithmica \textbf{72}(3): 734--757 (2015)

\bibitem{DKM18}
Dahn, C., Kriege, N. M., Mutzel, P.: A fixed-parameter algorithm for the Max-Cut problem on embedded 1-planar graphs.
In: Proceedings of IWOCA 2018, LNCS, vol. 10979, pp. 141--152 (2018)

\bibitem{Gab83}
Gabow, H. N.: An efficient reduction technique for degree-constrained subgraph and bidirected network flow problems.
In: Proceedings of STOC 1983, pp.~448--456 (1983)

\bibitem{GLV01}
Galluccio, A., Loebl, M., Vondr\'ak, J.: Optimization via enumeration: a new algorithm for the max cut problem. Mathematical Programming \textbf{90}(2): 273--290 (2001)

\bibitem{GS17} Gaspers, S., Sorkin, G. B.: Separate, Measure and Conquer: Faster Polynomial-Space Algorithms for Max 2-CSP and Counting Dominating Sets. ACM Trans. Algorithms \textbf{13}(4): 44:1-44:36 (2017)

\bibitem{GW95}
Goemans, M. X., Williamson, D. P: Improved approximation algorithms for maximum cut and satisfiability problem using semidefinite programming. Journal of the ACM \textbf{42}(6), 1115--1145 (1995)

\bibitem{Gur99}
Guruswami, V.: Maximum cut on line and total graphs.
Discrete Applied Mathematics \textbf{92}: 217--221 (1999)

\bibitem{Had75}
Hadlock, F.: Finding a maximum cut of a planar graph in polynomial time. SIAM Journal on Computing \textbf{4}(3): 221--255 (1975)

\bibitem{Has01}
H\r{a}stad, J.: Some optimal inapproximability results. J. ACM \textbf{48}(4): 798--859 (2001)

\bibitem{Karp72}
Karp, R. M.: Reducibility among combinatorial problems. In: Miller R.E., Thatcher J.W., Bohlinger J.D. (eds) Complexity of Computer Computations. The IBM Research Symposia Series. 85--103, Springer, Boston, MA (1972)

\bibitem{KKMO07}
Khot, S., Kindler, G., Mossel, E., O'Donnell, R.: Optimal inapproximability results for MAX-CUT and Other 2-Variable CSPs? SIAM J. Comput. \textbf{37}(1): 319--357 (2007)

\bibitem{LP12}
Liers, F., Pardella, G.: Partitioning planar graphs: a fast combinatorial approach for max-cut. 
Computational Optimization and Applications \textbf{51}(1), 323--344 (2012)

\bibitem{LT80}
Lipton, R. J., Tarjan, R. E.: Applications of a Planar Separator Theorem.
SIAM Journal on Computing \textbf{9}(3): 615--627 (1980) 

\bibitem{MSZ18}
Madathil, J., Saurabh, S., Zehavi, M.: Max-Cut Above Spanning Tree is Fixed-Parameter Tractable. In: Proceedings of CSR 2018. LNCS, vol.~10846, pp.~244--256, Springer (2018)

\bibitem{MR99}
Mahajan, M., Raman, V.: Parameterizing above Guaranteed Values: MaxSat and MaxCut. J. Algorithms \textbf{31}(2): 335-354 (1999)

\bibitem{Mor88}
Moret, B. M. E. : Planar NAE3SAT is in P. SIGACT News \textbf{19}(2): 51-54 (1988)

\bibitem{OD72}
Orlova, G. I., Dorfman: Finding the maximal cut in a graph.
Engineering Cyvernetics \textbf{10}(3), 502--506 (1972)

\bibitem{PPW16}
Pilipczuk, M. Pilipczuk, M., Wrochna: Edge Bipartization Faster then $2^k$.
In: Proceedings of IPEC 2016, LIPIcs, vol.~62, 26:1--26:13 (2016)

\bibitem{Poc16}
Pocai, R. V.: The Complexity of SIMPLE MAX-CIT on Comparability Graphs.
Electronic Notes in Discrete Mathematics \textbf{55}, 161--164 (2016)

\bibitem{RS07}
Raman, V., Saurabh, S.: Improved fixed parameter tractable algorithms for two ``edge'' problems: MAXCUT and MAXDAG. Inf. Process. Lett. \textbf{104}(2): 65-72 (2007)

\bibitem{Sch78}
Schaefer, T. J.: The complexity of satisfiability problems.
In: Proceedings of the Tenth Annual ACM Symposium on Theory of Computing. STOC '78.
pp.~216--226, ACM, New York (1978)

\bibitem{SWK90}
Shih, W.-K., Wu, S., Kuo, Y. S.: Unifying maximum cut and minimum cut of a planar graph. IEEE Transactions on Computers \textbf{39}(5), 694--697 (1990)

\bibitem{TSSW00}
Trevisan, L., Sorkin, G. B., Sudan, M., Williamson, D. P.: Gadgets, approximation, and linear programming. SIAM J. Comput. \textbf{29}(6): 2074--2097 (2000)

\bibitem{W05}
Williams, R.: A new algorithm for optimal 2-constraint satisfaction and its implications. Theor. Comput. Sci., \textbf{348}(2--3): 357--365 (2005)

\bibitem{Yan78}
Yannakakis, M.: Node- and edge-deletion NP-complete problems.
In: Proceedings of STOC 1978, pp.~253--264 (1978)

\end{thebibliography}
%

\end{document}